\documentclass[sigconf,nonacm]{acmart}
\makeatletter\@ACM@balancefalse\makeatother
\usepackage{tikz}
\usetikzlibrary{arrows.meta,positioning}
\usepackage{booktabs,tabularx,array}
\usepackage{enumitem}
\usepackage{float}
\newcommand{\PairColumns}{2}
\newcommand{\VerdictGroupRules}{\cmidrule(lr){2-3}\cmidrule(lr){4-5}\cmidrule(lr){6-7}}
\newcommand{\Verdict}[1]{\begingroup\edef\verdict{#1}\def\vpass{W}\def\vfail{X}\ifx\verdict\vpass$\checkmark$\else\ifx\verdict\vfail{\boldmath$\times$}\else?\fi\fi\endgroup}
\newcolumntype{L}[1]{>{\raggedright\arraybackslash}p{#1}}

\newcommand{\code}[1]{\texttt{#1}}
\newcommand{\RealLift}{4.4}
\newcommand{\RealLiftLo}{3.9}
\newcommand{\RealLiftHi}{5.1}
\newcommand{\FakeLift}{3.6}
\newcommand{\FakeLiftLo}{3.3}
\newcommand{\FakeLiftHi}{3.9}
\newcommand{\RepShare}{76\%}
\newcommand{\RepShareLo}{64\%}
\newcommand{\RepShareHi}{88\%}
\newcommand{\Gradient}{77}
\newcommand{\GradientLo}{65}
\newcommand{\GradientHi}{89}
\newcommand{\RealLiftRawLo}{4.0}
\newcommand{\RealLiftRawHi}{4.9}
\newcommand{\FakeLiftRawLo}{3.4}
\newcommand{\FakeLiftRawHi}{3.8}
\newcommand{\RepShareRawLo}{66\%}
\newcommand{\RepShareRawHi}{85\%}
\newcommand{\GradientRawLo}{68}
\newcommand{\GradientRawHi}{86}
\newcommand{\Support}{96\%}
\newcommand{\BalanceCount}{0.010}
\newcommand{\BalanceSlope}{0.007}
\newcommand{\FollowMin}{30}
\newcommand{\WashMin}{60}
\newcommand{\PreMin}{30}
\newcommand{\LandmarkMin}{5}
\newcommand{\DesignDraws}{40}
\newcommand{\BootDraws}{1000}
\newcommand{\InnerDraws}{10}
\newcommand{\JulyShare}{89\%}
\newcommand{\AugustShare}{66\%}
\newcommand{\SeptemberShare}{79\%}
\newcommand{\ActiveShare}{67\%}
\newcommand{\QuietShare}{\ensuremath{-}10\%}
\newcommand{\CoverageK}{1.30}
\newcommand{\WorstCoverage}{87\%}
\newcommand{\NewsJulsepShare}{55\%}
\newcommand{\StudyYear}{2026}
\newcommand{\HookBinMin}{5}
\newcommand{\AuditSearchDate}{October 16, 2026}
\newcommand{\AuditYears}{2022--2026}
\newcommand{\AcadEligible}{104}
\newcommand{\AcadUserTimed}{16}
\newcommand{\AcadCausal}{15}
\newcommand{\AcadPct}{15\%}
\newcommand{\TradeSources}{48}
\newcommand{\TradeUserChosen}{37}
\newcommand{\TradePct}{77\%}
\newcommand{\ClockHours}{6}
\newcommand{\CensoredShare}{110\%}
\newcommand{\CensoredShareLo}{82\%}
\newcommand{\CensoredShareHi}{153\%}
\newcommand{\IntervalLevel}{95\%}
\newcommand{\BoundaryGapUs}{1}
\newcommand{\BoundaryGappedShare}{75.7\%}
\newcommand{\BoundaryGappedShareLo}{66\%}
\newcommand{\BoundaryGappedShareHi}{86\%}
\newcommand{\ScopeFollowupResponses}{78.9\%}
\newcommand{\ScopeActiveResponses}{4.5\%}
\newcommand{\ScopeActiveContribution}{11.8\%}
\newcommand{\ScopeActiveResponseBrief}{4.5\%}
\newcommand{\ScopeActiveExcessBrief}{12\%}
\newcommand{\EventTypes}{16}
\newcommand{\SweepShareMin}{16\%}
\newcommand{\SweepShareMax}{76\%}
\newcommand{\MedianChained}{78\%}
\newcommand{\OwnShareMin}{7\%}
\newcommand{\OwnShareMax}{29\%}
\newcommand{\OracleCells}{8}
\newcommand{\BenchmarkWilsonMax}{0.365}
\newcommand{\BenchmarkMcseMax}{0.180}
\newcommand{\Tolerance}{0.25}
\newcommand{\ThinRaw}{83\%}
\newcommand{\ThinWide}{89\%}
\newcommand{\AiRichWeight}{0.5\%}
\newcommand{\ShopRichWeight}{near zero}
\newcommand{\NewsRichWeight}{near zero}
\newcommand{\ChainPooling}{0.98}
\newcommand{\ChainPoolingLo}{0.71}
\newcommand{\ChainPoolingHi}{1.24}
\newcommand{\ChainContinuation}{78\%}
\newcommand{\ChainContinuationLo}{69\%}
\newcommand{\ChainContinuationHi}{89\%}
\newcommand{\AnyPlacebo}{1}
\newcommand{\AiVisitHard}{1.15}
\newcommand{\ShopVisitHard}{1.58}
\newcommand{\NewsVisitHard}{2.99}
\newcommand{\LadderClockRate}{2.87}
\newcommand{\LadderClockRateLo}{2.40}
\newcommand{\LadderClockRateHi}{3.35}
\newcommand{\LadderClockShare}{75\%}
\newcommand{\LadderClockShareLo}{63\%}
\newcommand{\LadderClockShareHi}{89\%}
\newcommand{\LadderActiveRate}{1.05}
\newcommand{\LadderActiveRateLo}{0.62}
\newcommand{\LadderActiveRateHi}{1.51}
\newcommand{\LadderActiveShare}{31\%}
\newcommand{\LadderActiveShareLo}{10\%}
\newcommand{\LadderActiveShareHi}{61\%}
\newcommand{\LadderHistoryRate}{0.79}
\newcommand{\LadderHistoryRateLo}{0.36}
\newcommand{\LadderHistoryRateHi}{1.21}
\newcommand{\LadderHistoryShare}{7\%}
\newcommand{\LadderHistoryShareLo}{\ensuremath{-}25\%}
\newcommand{\LadderHistoryShareHi}{27\%}
\newcommand{\LadderSupport}{94\%}
\newcommand{\InjectedEffect}{0.5}
\newcommand{\AiExpertCalibrated}{2.87}
\newcommand{\AiExpertCalibratedLo}{1.57}
\newcommand{\AiExpertCalibratedHi}{6.19}
\newcommand{\AiExpertOverlap}{1.41}
\newcommand{\AiExpertOverlapLo}{0.00}
\newcommand{\AiExpertOverlapHi}{3.42}
\newcommand{\InfluenceBound}{$<5$}
\newcommand{\LadderHistoryUndefined}{1.3\%}
\newcommand{\DesignOracleCells}{30}
\newcommand{\HardCases}{18}
\newcommand{\StrengthAllDrInconclusive}{57}
\newcommand{\RichAiFwdSlope}{0.022}
\newcommand{\RichAiFwdSlopeLo}{\ensuremath{-}0.108}
\newcommand{\RichAiFwdSlopeHi}{0.161}
\newcommand{\RichAiRevSlope}{0.012}
\newcommand{\RichAiRevSlopeLo}{\ensuremath{-}0.157}
\newcommand{\RichAiRevSlopeHi}{0.163}
\newcommand{\RichShopFwdSlope}{0.248}
\newcommand{\RichShopFwdSlopeLo}{0.105}
\newcommand{\RichShopFwdSlopeHi}{0.416}
\newcommand{\RichShopRevSlope}{0.277}
\newcommand{\RichShopRevSlopeLo}{0.137}
\newcommand{\RichShopRevSlopeHi}{0.413}
\newcommand{\RichNewsFwdSlope}{0.207}
\newcommand{\RichNewsFwdSlopeLo}{0.043}
\newcommand{\RichNewsFwdSlopeHi}{0.409}
\newcommand{\RichNewsRevSlope}{0.217}
\newcommand{\RichNewsRevSlopeLo}{0.042}
\newcommand{\RichNewsRevSlopeHi}{0.405}
\newcommand{\RichSlopeDraws}{1000}
\newcommand{\TermAiFwdAccel}{\ensuremath{-}0.065}
\newcommand{\TermAiFwdSearch}{0.297}
\newcommand{\TermAiFwdRecency}{\ensuremath{-}0.007}
\newcommand{\TermAiFwdAge}{0.097}
\newcommand{\TermAiFwdError}{0.040}
\newcommand{\TermAiFwdErrorLo}{\ensuremath{-}0.056}
\newcommand{\TermAiFwdErrorHi}{0.160}
\newcommand{\TermAiFwdBeyond}{0.0\%}
\newcommand{\TermAiFwdCoefficientDecision}{W}
\newcommand{\TermAiFwdEvaluationDecision}{W}
\newcommand{\TermAiFwdResidual}{2.5\%}
\newcommand{\TermAiFwdResidualLo}{\ensuremath{-}3.4\%}
\newcommand{\TermAiFwdResidualHi}{9.8\%}
\newcommand{\TermAiRevAccel}{\ensuremath{-}0.052}
\newcommand{\TermAiRevSearch}{0.242}
\newcommand{\TermAiRevRecency}{\ensuremath{-}0.002}
\newcommand{\TermAiRevAge}{0.054}
\newcommand{\TermAiRevError}{0.042}
\newcommand{\TermAiRevErrorLo}{\ensuremath{-}0.041}
\newcommand{\TermAiRevErrorHi}{0.146}
\newcommand{\TermAiRevBeyond}{0.0\%}
\newcommand{\TermAiRevCoefficientDecision}{W}
\newcommand{\TermAiRevEvaluationDecision}{W}
\newcommand{\TermAiRevResidual}{2.6\%}
\newcommand{\TermAiRevResidualLo}{\ensuremath{-}2.6\%}
\newcommand{\TermAiRevResidualHi}{9.1\%}
\newcommand{\TermShopFwdAccel}{\ensuremath{-}0.111}
\newcommand{\TermShopFwdSearch}{0.466}
\newcommand{\TermShopFwdRecency}{0.872}
\newcommand{\TermShopFwdAge}{0.013}
\newcommand{\TermShopFwdError}{0.290}
\newcommand{\TermShopFwdErrorLo}{0.155}
\newcommand{\TermShopFwdErrorHi}{0.444}
\newcommand{\TermShopFwdBeyond}{69.8\%}
\newcommand{\TermShopFwdCoefficientDecision}{I}
\newcommand{\TermShopFwdEvaluationDecision}{I}
\newcommand{\TermShopRevAccel}{0.391}
\newcommand{\TermShopRevSearch}{0.397}
\newcommand{\TermShopRevRecency}{0.274}
\newcommand{\TermShopRevAge}{\ensuremath{-}0.221}
\newcommand{\TermShopRevError}{0.200}
\newcommand{\TermShopRevErrorLo}{0.090}
\newcommand{\TermShopRevErrorHi}{0.331}
\newcommand{\TermShopRevBeyond}{21.0\%}
\newcommand{\TermShopRevCoefficientDecision}{I}
\newcommand{\TermShopRevEvaluationDecision}{I}
\newcommand{\TermNewsFwdAccel}{0.089}
\newcommand{\TermNewsFwdSearch}{0.146}
\newcommand{\TermNewsFwdRecency}{0.610}
\newcommand{\TermNewsFwdAge}{0.052}
\newcommand{\TermNewsFwdError}{0.193}
\newcommand{\TermNewsFwdErrorLo}{0.038}
\newcommand{\TermNewsFwdErrorHi}{0.380}
\newcommand{\TermNewsFwdBeyond}{27.6\%}
\newcommand{\TermNewsFwdCoefficientDecision}{I}
\newcommand{\TermNewsFwdEvaluationDecision}{W}
\newcommand{\TermNewsRevAccel}{0.230}
\newcommand{\TermNewsRevSearch}{\ensuremath{-}0.002}
\newcommand{\TermNewsRevRecency}{0.396}
\newcommand{\TermNewsRevAge}{\ensuremath{-}0.268}
\newcommand{\TermNewsRevError}{0.100}
\newcommand{\TermNewsRevErrorLo}{\ensuremath{-}0.054}
\newcommand{\TermNewsRevErrorHi}{0.282}
\newcommand{\TermNewsRevBeyond}{4.8\%}
\newcommand{\TermNewsRevCoefficientDecision}{I}
\newcommand{\TermNewsRevEvaluationDecision}{W}
\newcommand{\TermShopAccelChange}{$-0.111$ to $0.391$}
\newcommand{\StrengthSweepPassSummary}{15 of 72}
\newcommand{\FloorExcessMin}{1.8}
\newcommand{\FloorExcessMax}{3.1}
\newcommand{\RealExcessMin}{3.2}
\newcommand{\RealExcessMax}{15.0}
\newcommand{\LadderOverstate}{3.6}
\newcommand{\LadderOverstateRawLo}{2.65}
\newcommand{\LadderOverstateRawHi}{7.28}
\newcommand{\LadderOverstateUnboundedDraws}{0.1\%}
\newcommand{\MultiscaleGapCutoffs}{15, 30, 60, 120, 300, 600, 1800, 3600}
\newcommand{\FutureFwdPastNullError}{0.507}
\newcommand{\FutureFwdAllNullError}{0.070}
\newcommand{\FutureFwdAllDurationError}{3.78}
\newcommand{\EOneCAiFwdNullDecision}{W}
\newcommand{\EOneCAiFwdOverlayDecision}{W}
\newcommand{\EOneCAiRevNullDecision}{W}
\newcommand{\EOneCAiRevOverlayDecision}{W}
\newcommand{\EOneCShopFwdNullDecision}{I}
\newcommand{\EOneCShopFwdOverlayDecision}{I}
\newcommand{\EOneCShopRevNullDecision}{I}
\newcommand{\EOneCShopRevOverlayDecision}{I}
\newcommand{\EOneCNewsFwdNullDecision}{I}
\newcommand{\EOneCNewsFwdOverlayDecision}{I}
\newcommand{\EOneCNewsRevNullDecision}{W}
\newcommand{\EOneCNewsRevOverlayDecision}{I}
\newcommand{\EOneDAiJointFwd}{0.039}
\newcommand{\EOneDAiJointRev}{0.054}
\newcommand{\EOneDShopJointFwd}{0.384}
\newcommand{\EOneDShopJointRev}{0.284}
\newcommand{\EOneDNewsJointFwd}{0.262}
\newcommand{\EOneDNewsJointRev}{0.135}
\newcommand{\EOneDAiFwdIntegrated}{0.040}
\newcommand{\EOneDAiFwdGap}{\ensuremath{-}0.001}
\newcommand{\EOneDAiFwdGapLo}{\ensuremath{-}0.040}
\newcommand{\EOneDAiFwdGapHi}{0.037}
\newcommand{\EOneDAiRevIntegrated}{0.042}
\newcommand{\EOneDAiRevGap}{0.012}
\newcommand{\EOneDAiRevGapLo}{\ensuremath{-}0.038}
\newcommand{\EOneDAiRevGapHi}{0.063}
\newcommand{\EOneDShopFwdIntegrated}{0.290}
\newcommand{\EOneDShopFwdGap}{0.094}
\newcommand{\EOneDShopFwdGapLo}{0.046}
\newcommand{\EOneDShopFwdGapHi}{0.142}
\newcommand{\EOneDShopRevIntegrated}{0.200}
\newcommand{\EOneDShopRevGap}{0.084}
\newcommand{\EOneDShopRevGapLo}{0.032}
\newcommand{\EOneDShopRevGapHi}{0.135}
\newcommand{\EOneDNewsFwdIntegrated}{0.193}
\newcommand{\EOneDNewsFwdGap}{0.069}
\newcommand{\EOneDNewsFwdGapLo}{0.027}
\newcommand{\EOneDNewsFwdGapHi}{0.111}
\newcommand{\EOneDNewsRevIntegrated}{0.100}
\newcommand{\EOneDNewsRevGap}{0.035}
\newcommand{\EOneDNewsRevGapLo}{\ensuremath{-}0.009}
\newcommand{\EOneDNewsRevGapHi}{0.079}
\newcommand{\EOneDRichToleranceRatio}{5}
\newcommand{\EOneEAiFailFwd}{74.2\%}
\newcommand{\EOneEAiFailRev}{82.8\%}
\newcommand{\EOneEAiRate}{50}
\newcommand{\EOneENewsFailFwd}{95.1\%}
\newcommand{\EOneENewsFailRev}{72.8\%}
\newcommand{\EOneENewsRate}{200}
\newcommand{\EOneEShopFailMin}{11.8\%}
\newcommand{\EOneEShopFailMax}{86.4\%}
\newcommand{\EOneEShopRate}{50}
\newcommand{\EOneEAiSaturationRate}{200}
\newcommand{\EOneERequiredDecisionShare}{80\%}
\newcommand{\ETwoKnownShare}{75.7\%}
\newcommand{\ETwoValidationThreshold}{93.0\%}
\newcommand{\ETwoHeadlineWideCoverage}{98.2\%}
\newcommand{\ETwoHeadlineRawCoverage}{93.2\%}
\newcommand{\ETwoSparseRawCoverage}{99.5\%}
\newcommand{\ETwoSparseWideCoverage}{100.0\%}
\newcommand{\ETwoSparseRawMcsePP}{0.35}
\newcommand{\ETwoSparseWideMcsePP}{0}
\newcommand{\GridAiFwdRandom}{W}
\newcommand{\GridAiFwdMatched}{W}
\newcommand{\GridAiFwdFit}{W}
\newcommand{\GridAiFwdRichMatch}{X}
\newcommand{\GridAiRevRandom}{W}
\newcommand{\GridAiRevMatched}{W}
\newcommand{\GridAiRevFit}{W}
\newcommand{\GridAiRevRichMatch}{X}
\newcommand{\GridShopFwdRandom}{W}
\newcommand{\GridShopFwdMatched}{W}
\newcommand{\GridShopFwdFit}{I}
\newcommand{\GridShopFwdRichMatch}{X}
\newcommand{\GridShopRevRandom}{W}
\newcommand{\GridShopRevMatched}{W}
\newcommand{\GridShopRevFit}{I}
\newcommand{\GridShopRevRichMatch}{X}
\newcommand{\GridNewsFwdRandom}{W}
\newcommand{\GridNewsFwdMatched}{W}
\newcommand{\GridNewsFwdFit}{I}
\newcommand{\GridNewsFwdRichMatch}{X}
\newcommand{\GridNewsRevRandom}{W}
\newcommand{\GridNewsRevMatched}{W}
\newcommand{\GridNewsRevFit}{W}
\newcommand{\GridNewsRevRichMatch}{X}
\newcommand{\PublicWikiShare}{113\%}
\newcommand{\PublicWikiShareLo}{104\%}
\newcommand{\PublicWikiShareHi}{122\%}
\newcommand{\PublicWikiHistory}{\ensuremath{-}0.32}
\newcommand{\PublicWikiHistoryLo}{\ensuremath{-}0.60}
\newcommand{\PublicWikiHistoryHi}{\ensuremath{-}0.04}
\newcommand{\PublicWikiQuiet}{1\%}
\newcommand{\PublicWikiQuietLo}{$-\!<\!1\%$}
\newcommand{\PublicWikiQuietHi}{3\%}
\newcommand{\PublicWikiPeople}{2{,}366}
\newcommand{\PublicWikiEvents}{8{,}809}
\newcommand{\PublicMovieShare}{121\%}
\newcommand{\PublicMovieShareLo}{117\%}
\newcommand{\PublicMovieShareHi}{124\%}
\newcommand{\PublicMovieHistory}{\ensuremath{-}0.12}
\newcommand{\PublicMovieHistoryLo}{\ensuremath{-}0.31}
\newcommand{\PublicMovieHistoryHi}{0.06}
\newcommand{\PublicMovieQuiet}{11\%}
\newcommand{\PublicMovieQuietLo}{8\%}
\newcommand{\PublicMovieQuietHi}{14\%}
\newcommand{\PublicMoviePeople}{6{,}196}
\newcommand{\PublicMovieEvents}{8{,}628}
\newcommand{\LadderGap}{0.73}
\newcommand{\LadderGapLo}{0.23}\newcommand{\LadderGapHi}{1.25}\newcommand{\LadderGapCalK}{1.09}\newcommand{\LadderOverstateCalLo}{2.6}\newcommand{\LadderOverstateCalHi}{7.3}\newcommand{\LadderRatioCovRaw}{96.5\%}\newcommand{\LadderShareCovRaw}{91.5\%}\newcommand{\LadderShareCovWide}{95.5\%}\newcommand{\LadderHistoryShareWideLo}{$-34\%$}\newcommand{\LadderHistoryShareWideHi}{34\%}\newcommand{\LadderExactRatio}{6.7}
\newcommand{\PrePostInertShare}{four fifths}
\newcommand{\PrePostRatio}{0.48}%
\newcommand{\JointAiFwdLo}{$-0.06$}\newcommand{\JointAiFwdHi}{0.17}\newcommand{\JointAiRevLo}{$-0.05$}\newcommand{\JointAiRevHi}{0.15}
\newcommand{\ScreenAgreeUser}{97\%}\newcommand{\ScreenKappaUser}{0.90}\newcommand{\ScreenAgreeCausal}{95\%}\newcommand{\ScreenKappaCausal}{0.86}\newcommand{\ScreenUserLo}{15\%}\newcommand{\ScreenUserHi}{19\%}\newcommand{\ScreenMatched}{68\%}\newcommand{\ScreenPrePost}{21\%}
\newcommand{\MUShareAi}{75\%}\newcommand{\MUShareAiLo}{62\%}\newcommand{\MUShareAiHi}{92\%}
\newcommand{\MUDidAi}{83\%}\newcommand{\MUDidAiLo}{65\%}\newcommand{\MUDidAiHi}{105\%}
\newcommand{\MUClockAi}{75\%}\newcommand{\MUShareShop}{44\%}\newcommand{\MUClockShop}{43\%}\newcommand{\MUShareNews}{55\%}\newcommand{\MUClockNews}{55\%}
\newcommand{\MUSecAi}{49\%}\newcommand{\MUSecAiLo}{11\%}\newcommand{\MUSecAiHi}{264\%}\newcommand{\MUSecShop}{14\%}\newcommand{\MUSecNews}{29\%}
\newcommand{\MUSupport}{about three quarters}
\newcommand{\HookCaption}{Search around eligible AI responses and inert timestamps where no response happened, against search at ordinary moments (the clock baseline). Both climb before time zero; against the clock baseline, inert timestamps reproduce \RepShare{} [\RepShareLo{}, \RepShareHi{}] of the excess in the 30-minute outcome window (gray; pre-registered). Shaded bands are pointwise 95\% person-bootstrap intervals ($\dagger\ddagger$). The spike in the five minutes before zero precedes the response and lies outside the outcome window. Search pageviews per hour from \HookBinMin{}-minute bins, July--September \StudyYear{} (Appendix~\ref{app:design}).}
\newcommand{\HeroCaption}{Inert timestamps reproduce much of the excess search after eligible AI responses and after shopping and news visits. Each bar is the real-event excess in clock-baseline units: tan is inert excess and blue is the difference. Percentages are reproduced shares with 95\% intervals, also drawn as whiskers at the split. February (Feb) rows are development data. The July--September AI row is the registered replication; the shopping and news rows are exploratory and use the category anchors of Figure~\ref{fig:sweep}.}
\newcommand{\SweepCaption}{Across \EventTypes{} event types, inert search excess spans a narrower range than real-event excess. Whiskers are 95\% intervals; the gray band spans the inert point estimates in clock-baseline units, and the right column gives reproduced shares (inert excess divided by real excess). The AI row is the registered replication; every row uses the same eligibility rule and matching (Appendix~\ref{app:sweep}).}
\newcommand{\StrengthCaption}{Matching error for a true effect of zero, as timing depends more on history the matching ignores. Strength zero uses matched history only; one is the stress-test rule. F: matching cells and timing rules trained on July and tested on August--September; R: the reverse. Gray band: acceptable error ($\pm\Tolerance$ clock-baseline units). Whiskers on the curves are 95\% intervals. Amber circles show the error at the strength fitted to real events, and the amber bars below give that strength's 95\% interval; these common-slope point errors sit inside the band, and the error curves leave the band near strength 0.3, inside the shopping and news intervals. The AI zero line is an oracle that knows the true rule (Appendix~\ref{app:strength}).}
\newcommand{\LearnedCaption}{Under the learned timing model, matching and model-based adjustment remain inconclusive for every source and train/test direction. Each row is one event type and train/test direction (F: July to August--September; R: the reverse) with the learned timing model's score at full scale ($s=1$); circles are matching and triangles model-based adjustment. Lines are 95\% intervals. Errors use clock-baseline units with gray tolerance band $\pm\Tolerance$; an interval inside the band passes, one wholly outside fails, and any other is inconclusive ($\dagger$; Appendix~\ref{app:strength}).}
\newcommand{\ReferenceCaption}{Each row compares the same AI responses with a different set of moments: ordinary moments of the day (top), moments with recent browsing (middle), and moments with matching recent history (bottom). Circles are real responses and squares are inert timestamps, in searches per hour. Both fall as the comparison tightens. Thin lines are raw 95\% intervals ($\dagger\ddagger$). The real-minus-inert gap (black bar) is the same in every row by construction, so the share that inert timestamps reproduce shrinks with real excess; real excess against the clock baseline is \LadderOverstate{} times that against history-matched moments.}
\newcommand{\HeatmapCaption}{Browsing and search are already high before AI responses, shopping visits, and news visits. For AI responses after an hour without one, with no browsing filter, search in the half hour before a response runs at 4.7 times its level at ordinary moments, higher than in the half hour after. Each cell spans 30 minutes; color is the log2 ratio of activity to its level at ordinary moments (random times of day, not inert timestamps), so 0 is no difference and 2 is four times as much.}
\setcopyright{none}
\renewcommand\footnotetextcopyrightpermission[1]{}
\AtBeginDocument{}
\begin{document}
\title{Event-Time Confounding Under Bursty Human Dynamics}
\author{Michael Iannelli}
\affiliation{\institution{Scrunch AI}\city{New York}\country{USA}}
\email{michael@scrunchai.com}
\author{Alan Ai}
\affiliation{\institution{Scrunch AI}\city{New York}\country{USA}}
\email{alan.ai@scrunchai.com}
\renewcommand{\shortauthors}{Iannelli and Ai}

\begin{abstract}
Web studies often measure what happens after a user action, such as a search after a chatbot answer. To credit an AI answer with later commercial activity, a study must separate what the answer caused from what the user would have done anyway. But people act in the middle of activity already underway, so the activity that drives what follows also chooses when the clock starts. We call this \emph{endogenous time zero}. Using an opt-in panel that links conversational-AI use, browsing, and search, we insert inert timestamps (moments that match recent activity but mark no event) and rerun the standard analysis. In a pre-registered replication, these timestamps, which by construction cannot cause anything, reproduce \RepShare{} of the search increase after AI responses that follow recent browsing. Other users matched on 30-day activity, the usual control, give the same answer, and inert timestamps show a similar search excess across \EventTypes{} kinds of Web events. Comparing with moments of matching recent history removes most of the increase, but no placebo audit we ran could confirm that this comparison is unbiased, and published studies almost never run one. We give a protocol for choosing comparisons and stress-testing estimators.
\end{abstract}

\maketitle

\section{Introduction}
Conversational AI assistants now sit early in many commercial journeys on the Web. Panel studies report more visits to brands' sites after AI answers recommend them~\cite{iannelli2026fptp}, and that AI users split activity differently across dialogue, search, and browsing~\cite{iannelli2026shape}. Brands, retailers, and publishers want to know how much downstream activity an AI answer creates. That requires separating what the answer caused from what would have happened anyway. This paper is about how Web research should evaluate such claims.

Suppose someone is planning a trip. Over twenty minutes they open three hotel pages, ask a chatbot which neighborhood is safest, and then run three Google searches. A standard event study credits those searches to the chatbot's answer. But the question came in the middle of a busy stretch of planning, and that momentum would have produced many of those searches anyway. The same goes for a store visit during a shopping trip or a news click during a morning read.

We call this problem \emph{endogenous time zero}. The moment the clock starts is chosen by the same activity that drives what comes next, so activity already underway is mistaken for the event's effect.

\begin{figure}[t]
\centering
\includegraphics[width=\columnwidth]{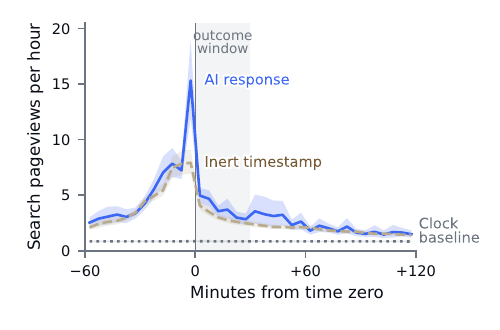}
\caption{\HookCaption}
\Description{Search rises before AI responses and inert timestamps matched on recent activity; random moments supply a much lower baseline.}
\label{fig:hook}
\end{figure}

Figure~\ref{fig:hook} shows the problem in our data. We focus on AI responses that come after recent browsing, which we call \emph{eligible} responses (\ScopeActiveResponseBrief{} of all responses). Search climbs in the half hour before these responses. We then placed timestamps at moments that look like the run-up to a response but where no response happened. We call these \emph{inert timestamps}, ghosts of events that never happened. Their search follows the real curve. In our pre-registered replication they reproduce \RepShare{} of the increase. When we instead compare each response with moments of similar recent activity, the increase falls to \LadderHistoryRate{} searches per hour, and the inert timestamps show almost none (Section~\ref{sec:comparison}).

The Web makes this problem common. Clicks, queries, purchases, and visits are the natural places to start a clock. Standard controls for who a person is do not help, because the problem lies in what the person is doing that minute.

Starting the clock at a user's own action is not uncommon in published work, and it is the norm in the industry reports we screened. We screened event-window studies of later behavior (Appendix~\ref{app:prevalence}). Of the academic studies in scope, \AcadPct{} started the clock at a user's action. Nearly all of them called the result an effect, and at most three checked it with a placebo. In a hand-picked set of industry sources, \TradeUserChosen{} of \TradeSources{} (\TradePct{}) defined exposure by a user's action. Academic studies mostly compared users with matched controls or their own past, industry sources with non-users. We ran the academic designs on our panel, and inert timestamps reproduce most of the increase under each (Section~\ref{sec:spot}).

\paragraph{Contributions.}
\begin{enumerate}[label=(\roman*),leftmargin=*,nosep]
\item \textbf{A placebo test for user-timed events.} A study cannot randomize when a person asks a chatbot, so we construct the placebo: we insert timestamps where nothing happened and measure how much of the apparent effect they reproduce. They reproduce about three quarters of the search increase after eligible AI responses: \MUShareAi{} [\MUShareAiLo{}, \MUShareAiHi{}] against other users matched on 30-day activity, the usual control, and \RepShare{} [\RepShareLo{}, \RepShareHi{}] against clock time in our pre-registered replication (Section~\ref{sec:spot}). Against moments with similar recent activity, most of the increase disappears (Section~\ref{sec:comparison}).\footnote{Brackets are \IntervalLevel{} intervals that resample whole people, widened as in Appendix~\ref{app:design} unless a coverage check showed no widening was needed; Monte Carlo intervals are labeled as such, and the appendix marks intervals that condition on fitted models or are unwidened.}

\item \textbf{Similar inert excess across the Web.} We find the same pattern for \EventTypes{} kinds of events, from shopping and news visits to webmail and coding documentation. Real-event excess varies almost fivefold across them, but inert excess moves far less. For the two smallest apparent effects (AI responses and news visits), inert timestamps reproduce \RepShare{} and \NewsJulsepShare{} of the excess (Section~\ref{sec:floor}). Two public logs, of Wikipedia editing and MovieLens tagging, show the same pattern (Appendix~\ref{app:public}).

\item \textbf{A stress test that standard placebo audits fail.} We simulated worlds where the true effect is zero but event timing depends on recent history that matching ignores, more strongly than real AI events show. Matching then reports a lift. A placebo audit on one realistic draw could not tell this world from one where matching works (Section~\ref{sec:cannot}). It caught the failure only with coarser matching and fifty to two hundred times as many inert timestamps as real events.

\end{enumerate}
Table~\ref{tab:protocol} turns these findings into a protocol, and Section~\ref{sec:protocol} works it through for AI responses.

\section{Endogenous time zero}\label{sec:related}
Figure~\ref{fig:dag} shows the structure. An ongoing state, which we call \emph{user momentum}, prompts the event and also keeps later activity going. Consider an episode of trip planning: the planning that prompts a chatbot question also produces the hotel searches that follow. The dashed arrow from the event to later activity is the effect a study wants. The path through the state adds an association even when that effect is zero.

\begin{figure}[H]
\centering
\begin{tikzpicture}[>=Stealth,font=\small,box/.style={draw,rounded corners,align=center,inner sep=4pt}]
\node[box] (b) at (0,1.4) {User momentum\\[-1pt]{\scriptsize(common cause)}};
\node[box] (a) at (-1.9,0) {User-timed event};
\node[box] (y) at (1.9,0) {Later activity};
\draw[->] (b)--(a);
\draw[->] (b)--(y);
\draw[->,dashed] (a)--node[midway,below,font=\scriptsize,yshift=-2pt]{effect of interest}(y);
\end{tikzpicture}
\caption{User momentum, such as an episode of trip planning, causes both the event and the activity that follows. The dashed arrow is the effect a study wants to measure.}
\Description{User momentum points to a user-timed event and to later activity, labeled as the common cause. A dashed event-to-activity arrow marks the effect of interest.}
\label{fig:dag}
\end{figure}

Endogenous time zero is this shared cause at the moment the clock starts. Bursty behavior makes it visible over minutes. Even a timestamp where nothing happened can predict what follows, because of where it was placed. We call the search increase after such timestamps \emph{inert excess} and measure it in Section~\ref{sec:spot}.

The problem has relatives in other fields. In job training, earnings dip just before people enroll, so the later rebound looks like a training effect~\cite{ashenfelter1978}. In medicine, a drug given for the first symptoms of an undiagnosed disease can look like the disease's cause~\cite{horwitz1980protopathic}. Target-trial emulation states the general rule: follow-up should start when a person becomes eligible and a treatment is assigned, or the study inflicts biases on itself~\cite{hernan2016target}; immortal-time bias is the best-known case~\cite{suissa2008immortal}. Human activity on the Web arrives in bursts~\cite{barabasi2005origin}, and self-exciting point processes model how one action raises the rate of the next~\cite{hawkes1971spectra}. Self-excitation lets momentum choose time zero. On the Web, split-door designs identify effects from auxiliary outcomes when a valid split exists~\cite{sharma2018split}; user-timed events rarely offer one. The closest prior work builds placebos, compares events with other times, or uses negative controls as audits.
\begin{itemize}[leftmargin=*,nosep]
\item Activity bias in advertising~\cite{lewis2011here} is measured against randomized control groups, and counterfactual bots script the user's choices~\cite{hosseinmardi2024bots}. Real users' timing allows neither, so we construct a placebo from inert timestamps. We also show that the bias depends on the comparison: inert timestamps reproduce most of the excess against clock time and little of it against moments with matching recent history (Section~\ref{sec:comparison}).
\item Case-crossover designs~\cite{maclure1991case}, self-controlled case series~\cite{whitaker2006sccs}, and the shuffle test~\cite{anagnostopoulos2008influence} compare events with other times, but not with times that match a person's recent activity. Risk-set matching~\cite{li2001balanced} and sequential trial emulation~\cite{hernan2008observational} do match on the recent past; inert timestamps test whether that is enough. Matching on the recent past is essential here, because inert excess is large at busy moments and absent at quiet ones (Section~\ref{sec:spot}).
\item Negative controls~\cite{lipsitch2010negative} are used as audits. We measure their power: with the analysis's own matching cells they stayed inconclusive even when matching was badly wrong (Section~\ref{sec:cannot}).
\end{itemize}
The recent debate over LLM use and scientific productivity raises the same question over months~\cite{renault2026comment,kusumegi2026stopping}. Our concern is the minutes around a single event.

Pre-event trends help but do not settle the question~\cite{roth2022pretest,freyaldenhoven2019pretrends,rambachan2023credible}. A rise before the event is a reason to look closer. A flat pre-period does not prove the comparison is fair.

\section{Testing with inert timestamps}\label{sec:spot}
\begin{figure*}[t]
\centering
\includegraphics[width=\textwidth]{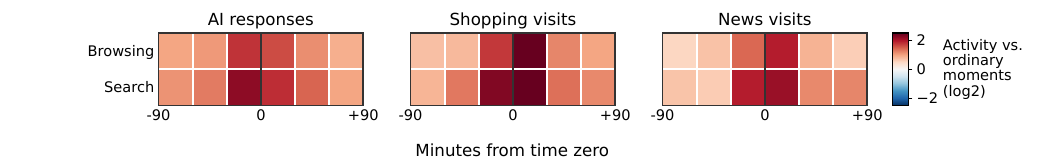}
\caption{\HeatmapCaption}
\Description{Three event panels show browsing and search rising before AI responses, shopping visits, and news visits, relative to the clock baseline, in the 90 minutes before and after each event.}
\label{fig:heatmap}
\end{figure*}

We insert timestamps where nothing happened, at moments that match a person's recent activity, and run the same event-window analysis. They obey the same rules as real events, including the hour since the last response, so none falls in a real response's follow-up window. An inserted timestamp cannot change behavior. Any search increase after it comes from where it sits in ongoing activity. Formally, an inert timestamp is a negative-control exposure~\cite{lipsitch2010negative} placed conditional on each person's recent history.

\paragraph{Data.}
Our data come from an opt-in commercial panel that records consenting users' web activity on one clock. It links page visits (domain and category), Google search pageviews, and assistant responses from ChatGPT, Gemini, Perplexity, and Claude. We leave out Google AI Overviews and AI Mode because they appear inside Google Search, the outcome we measure.\footnote{An AI Overview appears on the results page of the very search we count as the outcome. Treating it as an AI event would put the event inside its own outcome, so every such response would arrive with at least one search by construction.} The replication covers US and UK events from July through September \StudyYear{}. We do not report aggregate panel size, per-analysis user counts, or per-cell sample sizes; these are commercially sensitive. Effect sizes, confidence intervals, rates, and ratios are reported throughout. Raw panel data cannot be shared.

\paragraph{Construction.}
We use \emph{browsing} to mean pageviews outside search and AI. An eligible moment for AI has at least five browsing pageviews in the previous thirty minutes and no AI response in the previous hour, judged from earlier records only. We match each response to inert timestamps from the same person and month with similar recent browsing, a similar recent change in activity, and the same six-hour block of the day. The outcome is Google search pageviews over the next \FollowMin{} minutes. Unless we say otherwise, the comparison point is search at ordinary moments of the day, which we call the \emph{clock baseline}.

Let $R$ be the search rate after real events, $P$ after inert timestamps, and $Q$ at the comparison moments. Real excess is $R-Q$ and inert excess is $P-Q$. The \emph{reproduced share}, $(P-Q)/(R-Q)$, is the part of the apparent effect that inert timestamps also produce. Section~\ref{sec:comparison} reports excess in searches per hour. Elsewhere we divide excess by the clock-baseline rate so that event types with different search rates share one scale: an excess of one means one extra baseline's worth of search. Each person counts equally (details in Appendix~\ref{app:design}).

\paragraph{Pre-registered replication.}
In our pre-registered replication, inert timestamps reproduce \RepShare{} [\RepShareLo{}, \RepShareHi{}] of the excess search after eligible AI responses. Figure~\ref{fig:hero} adds the February development data and the shopping and news comparisons. We registered the analysis before seeing the replication outcomes~\cite{iannelli2026prereg}, and Appendix~\ref{app:design} lists the departures.

Most published studies compare event users with other users matched on longer-run activity. That design gives the same answer. We compared each AI response with other panelists at the same clock time who had similar browsing, search, and AI use over the previous 30 days. Thirty-day activity tells little about the last half hour, and these users search at the clock rate. Inert timestamps reproduce \MUShareAi{} [\MUShareAiLo{}, \MUShareAiHi{}] of the excess against them, and \MUDidAi{} [\MUDidAiLo{}, \MUDidAiHi{}] in the difference-in-differences version. Shopping and news agree (Appendix~\ref{app:matchedusers}).

\paragraph{Before/after designs.} Browsing and search are already high before AI responses, shopping visits, and news visits (Figure~\ref{fig:heatmap}, responses after an hour without one). For eligible responses, search is higher in the half hour before than in the half hour after, so a study that compares after with before would conclude that answers cut search roughly in half (to \PrePostRatio{} of its level before). The same comparison at matched inert timestamps shows about \PrePostInertShare{} of that drop, with no response at all. For AI, shopping, and news alike, the widened 95\% interval for the real-minus-inert difference includes zero. Because the before window falls inside the burst, a before/after design inherits the bias. Real responses also start from a higher level than matched inert timestamps, so this difference and the after-event gap in Section~\ref{sec:comparison} measure different things.

\begin{figure}[t]
\centering
\includegraphics[width=\columnwidth]{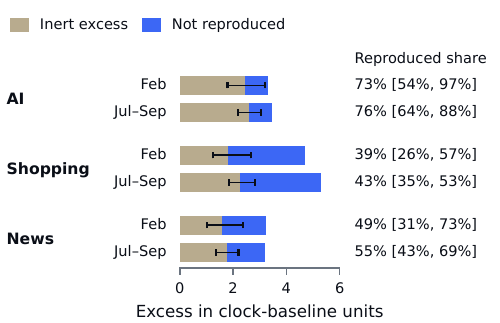}
\caption{\HeroCaption}
\Description{Excess search after real events is split into the portion reproduced by inert timestamps and a real-minus-inert remainder, for AI, shopping, and news in February and July--September.}
\label{fig:hero}
\end{figure}

Inert excess rises with how active the person was just before. Quiet moments have no browsing pageviews in the previous half hour, so they can be matched only on time of day. With that simpler matching, inert timestamps at eligible moments reproduce \ActiveShare{} of the excess. At quiet moments they reproduce none, and search even dips below the reference. This gap between eligible and quiet moments was the second pre-registered test, and it holds (Appendix Table~\ref{tab:rep}).

\section{Across Web event types}\label{sec:floor}
Inert excess is not specific to AI. We repeated the test for \EventTypes{} kinds of user-timed events, including shopping, news, coding documentation, webmail, reference sites, and search queries. Each time we required a gap since the last event of the same kind. The inert timestamps vary little: their excess stays between \FloorExcessMin{} and \FloorExcessMax{} clock-baseline units. The real events do not: their excess runs from \RealExcessMin{} to \RealExcessMax{} (Figure~\ref{fig:sweep}).

\begin{figure}[!t]
\centering
\includegraphics[width=\columnwidth]{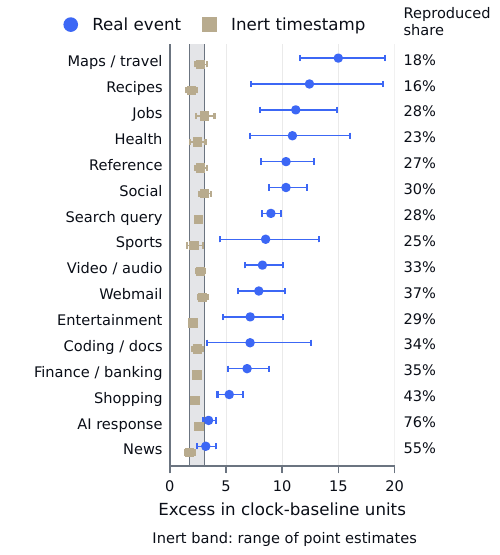}
\caption{\SweepCaption}
\Description{Matched inert search excess spans a narrower range than real-event excess across event types; the gray band spans inert point estimates and reproduced shares appear at right.}
\label{fig:sweep}
\end{figure}

So the same inert excess can be most of a small effect or a modest part of a large one. For every type, inert timestamps at eligible moments reproduce more than those at quiet moments (definitions in Appendix~\ref{app:sweep}).

\paragraph{Follow-ups and visits.}
Events also come in chains: a second chatbot question five minutes after the first, a third store page after the second. Without a gap rule, most events are follow-ups. For the median event type, \MedianChained{} come within an hour of another of the same kind, so each new start lands in the middle of the same busy stretch.

Grouping AI responses into visits does not fix this. Moments at the same point in other visits reproduce most of the follow-up search excess (Appendix~\ref{app:scope}).

\section{Better reference moments}\label{sec:comparison}
\begin{figure}[t]
\centering
\includegraphics[width=\columnwidth]{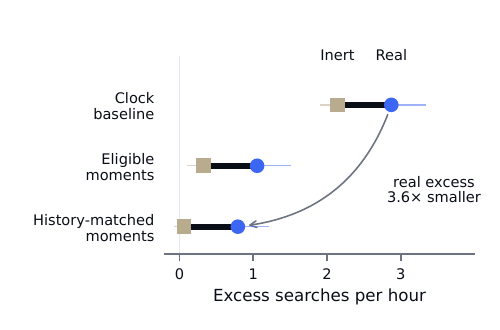}
\caption{\ReferenceCaption}
\Description{Changing from clock to eligible to history-matched reference moments moves real circles and inert squares together, leaving equal real-minus-inert gaps; horizontal intervals show raw conditional uncertainty.}
\label{fig:reference}\label{tab:real-ladder}
\end{figure}

Suppose we drop inert timestamps at random into moments when people are busy browsing. Search after them jumps relative to the whole day but not relative to other busy moments (Appendix~\ref{app:design-eval}).

For real responses we use \emph{history-matched moments}: eligible moments from the same person and month whose recent browsing, search, and activity change are closest to the event's. Figure~\ref{fig:reference} changes only the comparison moments. Real excess falls from \LadderClockRate{} searches per hour against ordinary moments to \LadderActiveRate{} against eligible moments and \LadderHistoryRate{} against history-matched ones. The inert share falls from about three quarters to \LadderActiveShare{} to \LadderHistoryShare{} [\LadderHistoryShareWideLo{}, \LadderHistoryShareWideHi{}]. The real-minus-inert gap stays \LadderGap{} [\LadderGapLo{}, \LadderGapHi{}] searches per hour in every row, because the same reference is subtracted from both, so comparing with history-matched moments and subtracting inert excess give nearly the same answer. Against clock time, the apparent effect is \LadderOverstate{} [\LadderOverstateCalLo{}, \LadderOverstateCalHi{}] times as large as against history-matched moments (Appendix~\ref{app:sweep}).

The same pattern appears in two public logs. Against clock time, inert timestamps reproduce \PublicWikiShare{} of the excess after Wikipedia talk-page edits and \PublicMovieShare{} after MovieLens tags. Against history-matched moments, the positive excess disappears in both (Appendix~\ref{app:public}).

\section{Limits of matching and placebo audits}\label{sec:cannot}
Matching on recent activity works when real timing depends only on the history it matches. It fails when timing depends strongly on history it leaves out, and the placebo audits an analyst would run are too imprecise to tell the two cases apart. The comparison moments of Section~\ref{sec:comparison} therefore help but are not sufficient. We simulated timing rules with a true effect of zero. The moments, outcomes, and weights stay fixed; only the rule that places the timestamps changes. Each rule is fit on one period and tested on the other, in both directions (F: July to August--September; R: the reverse).

We treat an error smaller than $\pm\Tolerance$ clock-baseline units as acceptable. Table~\ref{tab:term-errors} asks whether the method is biased on average, using expected timestamp counts, and marks each result $\checkmark$ (inside the band), {\boldmath$\times$} (outside it), or ? (inconclusive). An analyst sees only one random draw of inert timestamps, so we ask separately whether that one draw can reveal the bias. Gates are in Appendix~\ref{app:practitioner-audit}.

\subsection{A zero-effect stress test}
\paragraph{History that matching leaves out.}
First we made timing depend on details the matching cells leave unresolved: recent search, five-minute activity changes, time since the last browsing pageview, and session age. Figure~\ref{fig:strength} turns this extra dependence up from zero (matched variables only) through one (the fixed rule) and beyond. At strength one, matching fails for every event type in both directions (Table~\ref{tab:term-errors}). An oracle that knows the true timing rule stays at zero, which confirms the simulation works.

\begin{figure*}[t]
\centering
\includegraphics[width=0.8\textwidth]{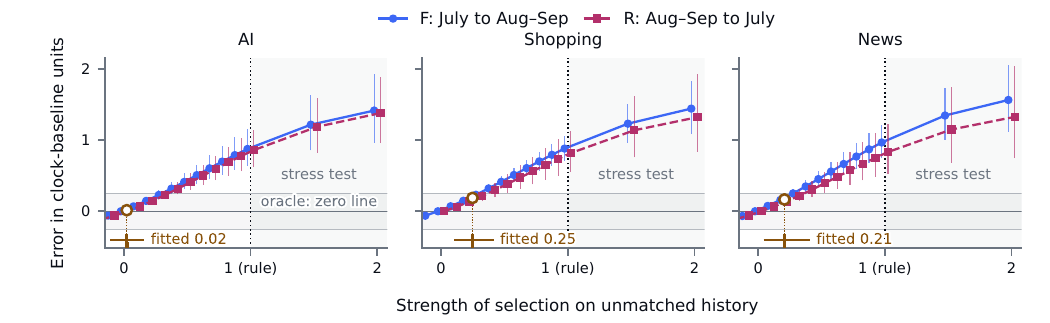}
\caption{\StrengthCaption}
\Description{Matching error rises with selection on unmatched history in AI, shopping and news. Solid lines show forward transfer, dashed lines reverse, and the gray band shows tolerance.}
\label{fig:strength}
\end{figure*}

\paragraph{Timing fitted to real events.}
Strength one is a stress test. When we fit the omitted terms to real events, AI sits near zero on the same axis, and shopping and news at strengths near 0.25 (Figure~\ref{fig:strength}; Appendix Table~\ref{tab:strength-slopes}). Matching's error leaves the tolerance band near strength 0.3. The shopping and news intervals reach about 0.4. For those two event types, the margin is thin. Holding the training fit fixed, matching under the fitted rule works for AI, is mixed for news, and is inconclusive for shopping (Table~\ref{tab:term-errors}). A more flexible model of real timing reaches no verdict for any event type (Appendix Figure~\ref{fig:learned}).

\begin{table}[b]
\centering\small\setlength{\tabcolsep}{2pt}
\caption{Does matching recover a zero effect on average? $\checkmark$ error within $\pm\Tolerance$ clock-baseline units, {\boldmath$\times$} outside, ? inconclusive, from expected timestamp counts with the training fit held fixed. F: train July, test August--September; R: the reverse. An analyst's audit on one random draw was inconclusive in every cell (Appendix~\ref{app:sampled}).}
\label{tab:term-errors}
\begin{tabular}{@{}l*{6}{>{\centering\arraybackslash}p{1.6em}}@{}}\toprule
Simulated timing rule & \multicolumn{\PairColumns}{c}{AI} & \multicolumn{\PairColumns}{c}{Shopping} & \multicolumn{\PairColumns}{c}{News}\\
\VerdictGroupRules
Train/test direction & F & R & F & R & F & R\\\midrule
Random within cells & \Verdict{\GridAiFwdRandom} & \Verdict{\GridAiRevRandom} & \Verdict{\GridShopFwdRandom} & \Verdict{\GridShopRevRandom} & \Verdict{\GridNewsFwdRandom} & \Verdict{\GridNewsRevRandom}\\
Uses matched history only & \Verdict{\GridAiFwdMatched} & \Verdict{\GridAiRevMatched} & \Verdict{\GridShopFwdMatched} & \Verdict{\GridShopRevMatched} & \Verdict{\GridNewsFwdMatched} & \Verdict{\GridNewsRevMatched}\\
Fitted to real events & \Verdict{\GridAiFwdFit} & \Verdict{\GridAiRevFit} & \Verdict{\GridShopFwdFit} & \Verdict{\GridShopRevFit} & \Verdict{\GridNewsFwdFit} & \Verdict{\GridNewsRevFit}\\
Uses unmatched history & \Verdict{\GridAiFwdRichMatch} & \Verdict{\GridAiRevRichMatch} & \Verdict{\GridShopFwdRichMatch} & \Verdict{\GridShopRevRichMatch} & \Verdict{\GridNewsFwdRichMatch} & \Verdict{\GridNewsRevRichMatch}\\
\bottomrule\end{tabular}
\end{table}

\subsection{Power of a single audit}
Placebo audits with the analysis's own matching cells stayed inconclusive even when matching was badly wrong. Their intervals were too wide to decide either way (Appendix~\ref{app:operational}). Coarser cells caught the failure in most draws for AI and news but unreliably for shopping, and only with inert timestamps at fifty to two hundred times the real event rate. No audit we ran showed that a comparison was unbiased.

\paragraph{Other fixes.}
None of the fixes we tried passed every rule. Matching on more history kept too few comparable moments, visit-level comparisons changed the question, and model-based adjustment left error or uncertainty (Appendix~\ref{app:expert}). An analysis that passes a no-effect check may still miss a real effect. Adjusting for post-event activity passes the forward no-effect check, yet fails badly when the event changes how long activity lasts (Appendix~\ref{app:duration}).

In the random-timestamp benchmark, the same timestamps give a biased comparison against clock time and an unbiased one against other eligible moments. A diagnostic that looks only at activity around the timestamps cannot tell these apart (Appendix~\ref{app:benchmark}).

\subsection{Limits of log-based checks}
\paragraph{Generated placebos.}
A timestamp drawn from recorded history carries no hidden information about the outcome, but a real event may reflect intent the logs never saw. Passing such a placebo therefore shows only that the analysis handles the rule that generated it.

\paragraph{Residual state.}
The bias grows with what the logs leave unknown about a person's state. If the logs make clear whether someone is in an active stretch, or if being active changes neither when they act nor how much they search, the bias disappears. How large it gets depends on how much doubt the logs leave and how strongly the state drives both. Appendix~\ref{prop:j3-bias} gives the formula and proof.

When the logs leave doubt about whether someone is in an active stretch, the event itself resolves part of that doubt. Longer bursts raise the total bias over a fixed window, other things equal (Appendix~\ref{prop:burst-persistence}). Using only pre-event records avoids counting the event's own consequences. It does not remove selection on a state the logs never saw.

\paragraph{Observational equivalence.}
Finally, imagine two worlds. In one, chatbot answers cause no extra searches; in the other, they do. We build both so that they produce exactly the same logs (Appendix~\ref{app:proofs}). Every audit computed from those logs, inert timestamps included, then behaves the same in both worlds. Only outside information, such as randomization or a defended assumption~\cite{miao2018proxy,tchetgen2024proximal}, can tell them apart.

\section{Practical guidance}\label{sec:protocol}
\begin{table}[b]
\caption{Protocol for user-timed event studies.}
\label{tab:protocol}
\small
\begin{tabularx}{\columnwidth}{@{}L{.80in}>{\raggedright\arraybackslash}X@{}}\toprule
Step & Action\\\midrule
Define & Name the event, the moments when it could have happened, the comparison moments, the outcome window, and how cases are weighted.\\\addlinespace[3pt]
Inspect & Plot activity before the event. Compare inert timestamps at active moments with those at quiet moments.\\\addlinespace[3pt]
Compare & Compare with moments defined from the past. For repeat events, say how earlier events are handled.\\\addlinespace[3pt]
Falsify & Add known synthetic effects to the outcomes. Place inert timestamps using history the analysis matches on and history it ignores. Raise the inert-timestamp rate until comparison moments run out, and repeat the draws. If you merge matching cells, audit the merged version. If you adjust for post-event activity, add effects that change how long the activity lasts.\\\addlinespace[3pt]
Report & Give the inert-timestamp share, with its interval, next to any lift. Say whether each check failed, passed, or was inconclusive, and which fitted models and design choices its intervals held fixed.\\\addlinespace[3pt]
Identify & Before a causal claim, randomize or defend an identifying assumption.\\
\bottomrule\end{tabularx}
\end{table}
Appendix Table~\ref{tab:taxonomy} lists, for each check, the failures it can detect and those it cannot.
\paragraph{Worked example: eligible AI responses.}
Take one illustrative response; the numbers below summarize the panel. The response is recorded at 9:12\,pm, after eight browsing pageviews and two Google searches in the previous half hour and no AI response in the previous hour.
\begin{description}[leftmargin=0pt,nosep,font=\normalfont\itshape]
\item[Define.] The event is the response. It could have happened at any eligible moment: at least five browsing pageviews in the previous thirty minutes and no response in the previous hour. The outcome is Google search pageviews over the next \FollowMin{} minutes, and each person counts equally.
\item[Inspect.] The response's inert timestamps come from the same person and month, in the same six-hour block, at moments with similar browsing in the previous half hour and a similar recent change in activity, with no response in the previous hour. They are not matched on the two searches. Search is already high before responses (Figure~\ref{fig:heatmap}), and under clock-only matching inert timestamps reproduce most of the excess at eligible moments and none at quiet ones.
\item[Compare.] Its history-matched moments are up to five of that person's eligible moments in the same month whose browsing over several windows, time since the last pageview, five-minute change, and recent search are closest, such as an earlier evening with seven pageviews and two searches. Against these, real excess falls to \LadderHistoryRate{} searches per hour and the inert share to \LadderHistoryShare{} [\LadderHistoryShareWideLo{}, \LadderHistoryShareWideHi{}] (exploratory).
\item[Falsify.] On average, matching passes when timing depends on matched variables or follows a rule fitted to real AI events, and fails when timing depends on unmatched history. One random draw of inert timestamps with the analysis's own cells is inconclusive.
\item[Report.] ``Against clock time, inert timestamps reproduce \RepShare{} [\RepShareLo{}, \RepShareHi{}] of the excess search after eligible AI responses (registered). Against history-matched moments, real excess is \LadderHistoryRate{} searches per hour (exploratory). Our placebo audit could not rule out selection on unmatched history.''
\item[Identify.] We make no causal claim; one would need randomization or a defended identifying assumption.
\end{description}

Ad experiments show how far observational attribution can sit from what a randomized test measures~\cite{gordon2019comparison,blake2015}; an experiment estimates the effect of what was randomized, so log each assignment and the moments when it could have happened.

\section*{Conclusion}
Event studies of user actions start the clock at the action, yet the momentum that prompts the action also drives what follows. In our panel, inert timestamps, placed where nothing happened, reproduce about three quarters of the search increase after AI responses that follow recent browsing, against clock time and matched users alike. The clock share is the highest in the event-type sweep (\SweepShareMin{} to \SweepShareMax{}). Comparing with moments of similar activity removes most of it, and a causal claim still needs randomization or a defended assumption. Event-window studies should report how event times were chosen, what they were compared with, and what happens when the event is replaced by an inert timestamp.

\section*{Limitations}
These analyses are observational. They do not identify real-event causal effects, and history-matched moments give a different comparison, not causal ground truth. No randomized or external variation in exposure was available to test the estimate. The main analyses use one opt-in AI/search panel; the public logs cover editing and rating, not AI or search. The panel analyses use overlapping categories, search-pageview proxies, and temporal visits not validated as semantic tasks. Eligible AI responses are a minority (\ScopeActiveResponseBrief{} of responses and \ScopeActiveExcessBrief{} of assessable search excess; these scope shares weight responses equally, unlike the person-weighted estimates), so the AI results need not extrapolate to all responses. Response timestamps may not mark exactly when an answer appeared or was read.

Only the AI replication is pre-registered. The registration disclosed earlier content analyses of the same months, and two departures affect its uncertainty: calibration intervals were approximated, and undefined ratio draws were kept as unbounded (Appendix~\ref{app:design}). Everything else is exploratory: the public-log replication, the reference comparison and its \LadderOverstate{}-fold ratio, the event-type sweep, the follow-up analyses, and the simulations. Our interval widening was calibrated for shares, not for that ratio (its raw interval covered \LadderRatioCovRaw{} in simulation), and shopping and news interval coverage is unvalidated. Simulated timing rules cannot exhaust hidden intent, missing data, or heterogeneous effects. Fitted-rule verdicts hold the training fit fixed; combining training and evaluation uncertainty leaves every fitted-rule cell inconclusive at the simultaneous level.

The matched-user comparison is exploratory, holds the matched-control means fixed in its bootstrap, and drops each person's first week. Adding recent browsing to that match lowers the AI share to \MUSecAi{} [\MUSecAiLo{}, \MUSecAiHi{}], mainly by removing most of the real excess too. History-matched moments match only part of the history-driven rise: in simulation, about half of the history-matched excess is history that matching misses, and against identical histories the clock-to-history ratio would be nearer \LadderExactRatio{} (Appendix~\ref{app:sweep}). History-matched intervals cover below nominal in the thin-reference benchmark, and the joint training-and-evaluation bootstrap inherits a widening factor calibrated for evaluation people only. Commercial restrictions withhold panel and cell counts, so matched-coverage and concentration diagnostics provide partial evidence about precision and generality. The literature screen describes the work we retrieved, not how often published conclusions are biased.

\section*{Ethics and human-subjects statement}
This is secondary analysis of previously collected, opt-in, de-identified web-browsing, search, and conversational-AI events; panelists consented to behavioral measurement under the panel provider's terms. We did not intervene or contact participants, and collected no new data. This paper's analyses use event timestamps and page categories, not page or conversation content. Results are reported only in aggregate; no raw conversations, searches, URLs, or sample-size statistics that could expose panel composition are released. Linking AI use to browsing and search for the same person is more sensitive than either stream alone, because pseudonymous browsing can identify people~\cite{su2017deanonymizing}.

\label{main:end}
\bibliographystyle{ACM-Reference-Format}
\bibliography{references}
\appendix
\section{Replication and uncertainty}\label{app:design}
Different analyses use different uncertainty procedures, and the marks say which: $\dagger$ marks conditioning on fitted designs or models, and $\ddagger$ marks raw (unwidened) intervals. Figure~\ref{fig:hook} shows pointwise person-bootstrap intervals over sampled designs ($\dagger\ddagger$). The exploratory reference comparison holds training bins fixed ($\dagger$).

\paragraph{Design.}
The replication uses July--September \StudyYear{}; development used February. An eligible AI response has no earlier response in the previous \WashMin{} minutes and at least $\LandmarkMin$ browsing pageviews in the previous \PreMin{} minutes. Quiet moments have no browsing in that window, and intermediate moments have one to four pageviews. Matching groups moments by person-month, $\ClockHours$-hour clock block, and two counts of browsing $N_o$ over minute intervals:
\[
 c_t=N_o(t-30,t),\quad s_t=N_o(t-15,t)-N_o(t-30,t-15).
\]
Bins split the real events' values of $c_t$ and $s_t$ at their terciles, with ties kept together and outliers placed in the end bins. The outcome counts Google search pageviews in $[t,t+\FollowMin)$. Under clock-only matching, inert timestamps reproduce \ActiveShare{} of the excess at eligible moments and \QuietShare{} at quiet ones, with no positive inert excess at quiet moments.

\paragraph{Weighting and shares.}
Each person counts equally: we average events within cells, cells within person-months, and months within people. A cell needs both arms and a positive clock baseline; zero outcomes stay in. Lift ratios divide event-window rates by the reference rate ($R/B$ and $P/B$, where one means no lift). The reproduced share is $(P-Q)/(R-Q)$, with $Q=B$ for the clock baseline, computed from averaged rates rather than as a mean of person ratios. Figure~\ref{fig:hero} uses category gaps of \WashMin{} minutes; its shopping and news intervals have no coverage check.

\paragraph{Bootstrap.}
Each bootstrap copy resamples whole people and rebuilds candidates, bins, support, and baselines. The point estimate averages \DesignDraws{} designs, and each of \BootDraws{} person resamples averages \InnerDraws{} fresh designs. When real excess is not positive, the share is undefined; we keep those draws as infinite endpoints with the sign of the inert excess. Nominal \IntervalLevel{} percentile endpoints are then widened about the point by \CoverageK{}.

\begin{table}[!ht]
\centering\footnotesize\setlength{\tabcolsep}{3pt}
\caption{Inert timestamps reproduce most of the search excess after eligible AI responses. Lift ratios have null one. The eligible-minus-quiet gradient uses percentage points and clock-only matching (\ActiveShare{} eligible versus \QuietShare{} quiet), unlike the share row.}
\label{tab:rep}
\begin{tabular}{@{}lrrr@{}}\toprule
Quantity & Point & Raw $\ddagger$ & Widened\\\midrule
Real lift ratio & \RealLift & [\RealLiftRawLo{}, \RealLiftRawHi{}] & [\RealLiftLo{}, \RealLiftHi{}]\\
Inert lift ratio & \FakeLift & [\FakeLiftRawLo{}, \FakeLiftRawHi{}] & [\FakeLiftLo{}, \FakeLiftHi{}]\\
Reproduced share & \RepShare & [\RepShareRawLo{}, \RepShareRawHi{}] & [\RepShareLo{}, \RepShareHi{}]\\
Eligible minus quiet & \Gradient & [\GradientRawLo{}, \GradientRawHi{}] & [\GradientLo{}, \GradientHi{}]\\
\bottomrule\end{tabular}
\end{table}

\paragraph{Interval widening.}
Raw bootstrap intervals cover too rarely in simulations that vary count ties, bin counts, and outcome tails. For cell coverage $c$, the widening factor is
\[
k=\lceil100\max_c\{1.959964/\Phi^{-1}((1+c)/2)\}\rceil/100,
\]
where $\Phi$ is the standard normal distribution function. The worst case, tied counts with ten bins, sets the factor; its raw coverage is \WorstCoverage{}.
A second check uses a simulated panel with a known share (\ETwoKnownShare{}). It resamples people and clustered months, rebuilds candidates, bins, and support, and adds person-level noise to each arm's denominator. In replication-like conditions, widened and raw coverage are \ETwoHeadlineWideCoverage{} and \ETwoHeadlineRawCoverage{}, so widening by \CoverageK{} meets the predeclared \ETwoValidationThreshold{} threshold.
With sparse references, raw and widened coverage are \ETwoSparseRawCoverage{} and \ETwoSparseWideCoverage{} (Monte Carlo standard errors \ETwoSparseRawMcsePP{} and \ETwoSparseWideMcsePP{} points).

\paragraph{Balance and influence.}
Matching keeps \Support{} of the target weight, with standardized mean differences of \BalanceCount{} for browsing count and \BalanceSlope{} for activity change. Dropping the most influential person moves the share by \InfluenceBound{} percentage points; monthly shares are \JulyShare{}, \AugustShare{}, and \SeptemberShare{}. Holm's correction covers both registered one-sided tests: positive inert elevation and the activity gradient.

\paragraph{Departures from registration.}
Registration revision \code{f8fb5f42}~\cite{iannelli2026prereg} preceded the outcomes and disclosed earlier content work on the same months. The implementation departs from it in six ways. Denominators leave out padding from adjacent months, while spans from first to last pageview keep history and follow-up across month boundaries. Activity classes get separate support. Replicate-level calibration intervals were unavailable, so we used the cell-coverage approximation above, fixed before outcomes. Undefined ratios keep infinite endpoints. Exploratory shopping and news first-visit anchors were registered but undefined; after outcomes we required no same-category visit in the previous hour. The abstract reports the registered share, not the H1 inert elevation (\FakeLift{}-fold).

\paragraph{Prevalence screen.}\label{app:prevalence}
One coder searched OpenAlex, conference venues, SSRN, and arXiv through \AuditSearchDate{} for \AuditYears{} observational event-window studies of later behavior or content. Of \AcadEligible{} eligible studies, \AcadUserTimed{} started the clock at a user's own action, \AcadCausal{} of those used causal language, and \AnyPlacebo{} reported a placebo test without describing its methods. Of \TradeSources{} hand-picked industry sources, \TradeUserChosen{} defined exposure or cohorts by a user's action; defenses we could not verify are coded unclear.
A second coder, an AI agent working from written coding rules, recoded the academic studies with partial blinding (aggregate tallies were visible, per-study labels were not). The coders agreed on \ScreenAgreeUser{} of user-timing calls (Cohen's $\kappa=\ScreenKappaUser{}$) and \ScreenAgreeCausal{} of causal-language calls ($\kappa=\ScreenKappaCausal{}$). The second coder found user timing in \ScreenUserLo{} to \ScreenUserHi{} of eligible studies, depending on four borderline cases, and placebo tests in two or three, one of them with methods described. Among user-timed studies, \ScreenMatched{} compared users with matched, weighted, or synthetic controls and \ScreenPrePost{} with their own earlier behavior; none used a clock baseline.

\section{Event definitions and reference comparisons}\label{app:sweep}
\paragraph{Event definitions.}
The sweep compares recorded event types, not validated tasks (Figure~\ref{fig:sweep}). Because categories and people overlap, we do not test whether types differ from one another. A non-AI event is a visit after a category-specific gap. Own-category outcomes leave out the triggering visit, and search-query outcomes count searches strictly after the trigger. Person resampling holds designs, candidates, and support fixed, and the sweep's intervals have no coverage check. Own-category reproduced shares span \OwnShareMin{}--\OwnShareMax{}, leaving out search's outcome on itself.

\paragraph{Heatmap construction.}
Figure~\ref{fig:heatmap} uses July--September \StudyYear{} events after a \WashMin{}-minute gap since the last event of the same type, with no browsing filter, so it is not limited to eligible responses. Colors show the log ratio of mean activity to unrestricted clock moments. Trigger visits count in browsing.

\paragraph{Scope and follow-ups.}\label{app:scope}
Eligible AI responses make up \ScopeActiveResponses{} of responses and \ScopeActiveContribution{} of assessable search excess; follow-ups within an hour of another response make up \ScopeFollowupResponses{} of responses. Quiet and intermediate classes match on clock block only, and spans we cannot assess add no excess.
Visits group responses with gaps under an hour. To test follow-ups, we compare each response with moments from other visits of the same person-month at the same elapsed time, matched on the number of earlier responses ($1,2,3$--$5,6$--$10,11+$), on browsing in the visit's first thirty minutes (zero, one to four, at least five), and on six-hour block. These moments reproduce \ChainContinuation{} [\ChainContinuationLo, \ChainContinuationHi]$\dagger\ddagger$ of follow-up excess. The ratio of all-response to first-response excess, \ChainPooling{} [\ChainPoolingLo, \ChainPoolingHi]$\dagger\ddagger$, gives no sign that pooling follow-ups inflates the estimate.

\paragraph{Exploratory reference comparison.}
The reference comparison changes only the comparison moments $Q$. It holds the real and inert rates $R$ and $P$, common support (\LadderSupport{}), and all weights fixed (Figure~\ref{fig:reference}; Table~\ref{tab:reference-comparison}). Columns average $R-Q$, $P-Q$, and $(P-Q)/(R-Q)$ within designs, so a share need not equal the ratio of displayed rates. Training bins are frozen on July people and evaluated on different people.
History-matched moments average up to five nearest eligible moments per real event, within the same cell and person-month, with replacement and a caliper of one on scaled root-mean-squared distance. Distance uses browsing counts over the previous 5, 15, 30, and 60 minutes, time since the last browsing pageview, the five-minute change in browsing, and search counts over the previous 5, 30, and 60 minutes. This comparison's clock-row share (\LadderClockShare{}) is computed on a different support and averaging than the registered replication's, so the two differ slightly.
The clock-to-history ratio divides mean clock excess by mean history-matched excess. Its paired \IntervalLevel{} interval is [\LadderOverstateRawLo{}, \LadderOverstateRawHi{}]$\dagger\ddagger$. A known-truth simulation built on the ladder (200 repetitions) checks coverage. The raw ratio interval covered \LadderRatioCovRaw{} of the matching target, so it needs no widening. The history-matched share covered \LadderShareCovRaw{} raw and \LadderShareCovWide{} widened, giving [\LadderHistoryShareWideLo{}, \LadderHistoryShareWideHi{}]$\dagger$. Against identical histories coverage fails, because about half the history-matched excess is history that nearest-neighbor matching misses (ratio \LadderExactRatio{} against identical histories). Draws with a nonpositive denominator (\LadderOverstateUnboundedDraws{}) count as an unbounded ratio. The gap interval uses its own calibrated factor, \LadderGapCalK{}.

\paragraph{Matched other users.}\label{app:matchedusers}
We fixed the matching variables before computing any outcome. Controls are other panelists within five minutes of the event's clock time, inside their observed span, with no same-type event from an hour before to half an hour after. They share the event's stratum: tercile of browsing and of search over the previous 30 sampled days, crossed with same-type use (none or a tercile). A person needs seven days of history; this rule and control availability keep \MUSupport{} of eligible anchors. On the same anchors the clock share is \MUClockAi{} for AI, \MUClockShop{} for shopping, and \MUClockNews{} for news; the matched-user shares are \MUShareAi{}, \MUShareShop{}, and \MUShareNews{}. Adding the 30-minute browsing bin lowers the shares (AI \MUSecAi{}, shopping \MUSecShop{}, news \MUSecNews{}), but the real AI excess falls with it and its interval covers zero. The registered widening is applied, not recalibrated.

\paragraph{Public logs.}\label{app:public}
We reran the analysis on English Wikipedia (Wikimedia \texttt{mediawiki\_history} dump of September \StudyYear{}, \url{https://dumps.wikimedia.org/other/mediawiki_history/}; registered, non-bot editors; July--September) and GroupLens MovieLens 32M (\url{https://files.grouplens.org/datasets/movielens/ml-32m.zip}~\cite{harper2015}), with talk-page edits and tags as events. The analyzed samples are \PublicWikiEvents{} eligible talk-page edits by \PublicWikiPeople{} editors and \PublicMovieEvents{} eligible tags by \PublicMoviePeople{} users. Splitting article edits by page ID and ratings by movie ID gives activity and outcome streams that share each person's bursts. Definitions follow Sections~\ref{sec:spot} and~\ref{sec:comparison}; outcome windows start just after the event, since MovieLens timestamps are whole seconds. Against clock time, inert timestamps reproduce \PublicWikiShare{} [\PublicWikiShareLo{}, \PublicWikiShareHi{}] (Wikipedia) and \PublicMovieShare{} [\PublicMovieShareLo{}, \PublicMovieShareHi{}] (MovieLens) of the excess; shares above 100\% mean real events are followed by slightly less activity. Against history-matched moments, real excess is \PublicWikiHistory{} [\PublicWikiHistoryLo{}, \PublicWikiHistoryHi{}] and \PublicMovieHistory{} [\PublicMovieHistoryLo{}, \PublicMovieHistoryHi{}] outcome records per hour. At quiet moments, inert timestamps reproduce \PublicWikiQuiet{} [\PublicWikiQuietLo{}, \PublicWikiQuietHi{}] and \PublicMovieQuiet{} [\PublicMovieQuietLo{}, \PublicMovieQuietHi{}]. All intervals are $\dagger\ddagger$.

\paragraph{Reference construction.}
Candidate times are drawn uniformly between the first and last pageviews of each usable person-month: $4000$ candidates before eligibility filtering and $400$ clock moments per design. The inert arm $P$ averages eligible candidates in the real event's cell; the eligible reference drops the cell restriction within the person-month. Recency uses cutoffs of $1,5,15,30,60$ minutes plus none within an hour, and five-minute activity change is twice the five-minute count minus the ten-minute count. Features are log-scaled, $\operatorname{sign}(x)\log(1+|x|)$, and divided by July-training interquartile ranges floored at one. Session-history matching adds session-age cutoffs of $5,15,30,60,120$ minutes and a first-observed-session flag; gaps over thirty minutes start a new session. Matching never uses future outcomes or session endpoints.

\begin{table}[!ht]
\centering\footnotesize\setlength{\tabcolsep}{3pt}
\caption{More active reference moments reduce both real excess and reproduced share in the exploratory AI comparison. Raw intervals are shown for the paired references ($\dagger\ddagger$).}
\label{tab:reference-comparison}
\begin{tabular}{@{}lrr@{}}\toprule
Reference & Real excess (searches/hour)$\dagger\ddagger$ & Reproduced share$\dagger\ddagger$\\\midrule
Clock baseline & \LadderClockRate{} [\LadderClockRateLo{}, \LadderClockRateHi{}] & \LadderClockShare{} [\LadderClockShareLo{}, \LadderClockShareHi{}]\\
Eligible & \LadderActiveRate{} [\LadderActiveRateLo{}, \LadderActiveRateHi{}] & \LadderActiveShare{} [\LadderActiveShareLo{}, \LadderActiveShareHi{}]\\
History-matched & \LadderHistoryRate{} [\LadderHistoryRateLo{}, \LadderHistoryRateHi{}] & \LadderHistoryShare{} [\LadderHistoryShareLo{}, \LadderHistoryShareHi{}]\\
\bottomrule\end{tabular}
\end{table}
In a benchmark with thin references, intervals for comparisons within eligible moments cover \ThinRaw{} raw and \ThinWide{} widened, so sparse matched coverage can undercut a valid comparison. History-matched shares are undefined in \LadderHistoryUndefined{} of draws, and their point estimate cannot establish equivalence.

\paragraph{Response timestamps.}\label{app:boundary}
Starting the window \BoundaryGapUs{} microsecond before the response gives a raw share of \BoundaryGappedShare{} [\BoundaryGappedShareLo, \BoundaryGappedShareHi]$\dagger\ddagger$. Ending each window at the next response gives a share of \CensoredShare{} [\CensoredShareLo{}, \CensoredShareHi{}]$\dagger$; it exceeds one hundred percent because real windows contain later responses more often than inert windows.

\section{Timing rules and estimator evaluation}\label{app:benchmark}\label{app:design-eval}
\begin{table}[!ht]
\caption{Inert-timestamp checks: the failures each one detects and those it leaves undetected. Verdicts are integrated (expected-count) unless marked. Any generated placebo tests only the rule that generated it.}
\label{tab:taxonomy}
\footnotesize\setlength{\tabcolsep}{3pt}
\begin{tabularx}{\columnwidth}{@{}L{.95in}>{\raggedright\arraybackslash}X>{\raggedright\arraybackslash}X@{}}\toprule
Check & Detects & Leaves undetected\\\midrule
Inert timestamps at random eligible moments & Bias in comparisons against clock time & Nonrandom timing within eligible moments\\
Inert timestamps that depend on matched variables & Matching errors under that rule ($\checkmark$ in all six cells) & Timing that depends on unmatched variables\\
Inert timestamps that depend on richer history & Matching when timing depends on unmatched history ({\boldmath$\times$} in all six cells at strength one) & Whether one sampled audit sees the failure (audits with the analysis's own cells were inconclusive)\\
Injected constant effect & Estimators that shrink every effect toward zero & Failure when the event changes duration\\
Injected duration-changing effect & Post-event adjustment that passes the null check (forward direction) & Whether recorded history explains real timing\\
\bottomrule\end{tabularx}
\end{table}
Every simulation fixes the moments, outcomes $Y_j$, baseline $B$, and weights $w_j$ before placing inert timestamps, and averages as in Appendix~\ref{app:design}. Candidate moments sit on a five-minute grid that keeps a day of history and stops an hour before the final pageview; the replication's gap and browsing rules apply, and future search is never used. Given preceding history $H_j$ and independent uniforms $U_j$, an inert timestamp lands at moment $j$ when $M_j=\mathrm I\{U_j<e_j\}$, where
\[
 e_j=\operatorname{clip}_{[10^{-6},1-10^{-6}]}\operatorname{expit}(\alpha+f(H_j)),
 \qquad \sum_{j\in\mathrm{train}}w_je_j=r.
\]
Here $r$ is the weighted real-event frequency among training people, $\operatorname{expit}(x)=1/(1+\exp(-x))$, and bisection on $[-40,40]$ sets $\alpha$. Random inert timestamps use $e_j=r$.

\paragraph{Timing rules.}
The fixed rules use the browsing-count bin $q(c_t)$, clock block $b_t$, acceleration $a_t=2N_o(t-5,t)-N_o(t-10,t)$, time since the last browsing pageview $d_t$, session age $v_t$, and the search count $N_s$. The matched-variable score and the richer-history score, which adds unmatched history, are
\begin{align*}
 f_C&=q(c_t)+0.25b_t,\\
 f_H&=f_C+1.25\mathrm I\{a_t>0\}+1.25\mathrm I\{N_s(t-30,t)>0\}\\
 &\quad+1.25\mathrm I\{d_t\leq2\}-0.75\mathrm I\{v_t>30\}.
\end{align*}
The learned timing model $T$ predicts event timing from a richer history: counts of browsing and search over $5,15,30,60,120,240,1440$ minutes, counts of the focal event over $60,240,1440$ minutes, activity change, hour, weekday, exponentially smoothed counts, acceleration, session age, and the age of the latest pageview after a thirty-minute gap, all strictly before the moment. A withheld-input rule $V$ copies $T$'s probabilities but hides the richer inputs from the analysis and drops activity change from matching. None of these rules estimates how strongly real events are selected.
Model-based adjustment uses unweighted boosting with the same settings across up to five held-out person groups, with outcome models fit by squared error on control moments. Both models are refit for each timestamp draw and held fixed across person resamples.

\paragraph{Estimators and oracle.}
The oracle knows the true assignment probabilities and uses fixed denominators. With $\pi=\sum_jw_je_j$, it is
\[
 \widehat\tau=\pi^{-1}\sum_jw_j
 \left\{M_j-\frac{(1-M_j)e_j}{1-e_j}\right\}Y_j,
\]
whose expectation over inert timestamps is zero; versions with random denominators lose this property. Matching with known probabilities averages $e_j$ within person-month, count, activity-change, and clock cells. Direct matching needs no probabilities: it weights each cell's timestamp-minus-control contrast by the cell's timestamp mass. The learner uses fitted probabilities and the residual $Y_j-\hat m_0(H_j)$. Dividing contrasts by $B$ gives clock-baseline units. Transfers train and evaluate on different people. Outside training, \OracleCells{} sampled-assignment oracle checks are inconclusive and \DesignOracleCells{} integrated checks pass.

\paragraph{Decision rules.}\label{app:operational}
A result is within tolerance (W; $\checkmark$ in Table~\ref{tab:term-errors}) when its whole interval lies in $[-\Tolerance,\Tolerance]$ and fails (X; {\boldmath$\times$}) when the interval lies wholly outside; anything else is inconclusive (I; ?). Both verdicts also need precision: a draw-mean Monte Carlo standard error of at most $0.02$, an undefined share of at most $0.025$, and finite endpoints. Table~\ref{tab:term-errors} uses simultaneous evaluation-person intervals for the random, matched-variable, and fitted rules, and marginal widened intervals for the strength-one richer-history row, holding training bins, coefficients, and models fixed. Simultaneous intervals widen by \CoverageK{} per arm.

\paragraph{Sampled audits.}\label{app:sampled}
A real audit sees one random draw of timestamps, and its errors can differ from the integrated ones. Under random and matched-variable rules, realized errors stay near zero; under the fitted rule they move away from their integrated values (Table~\ref{tab:realized-gap}), because missing arms change timestamp weight and cell mix. Under the fixed richer-history rule, realized errors average \EOneDRichToleranceRatio{} times the tolerance, yet the audits stay inconclusive.

\begin{table}[!ht]
\centering\footnotesize\setlength{\tabcolsep}{3pt}
\caption{Fitted-rule realized errors can differ from their integrated benchmarks in clock-baseline units. The gap uses different realized support; brackets are simultaneous Monte Carlo intervals over timestamp draws ($\dagger\ddagger$), excluding person-sampling and training-fit uncertainty.}
\label{tab:realized-gap}
\begin{tabular}{@{}lrrr@{}}\toprule
Source / transfer & Integrated & Realized & Gap [MC interval]$\dagger\ddagger$\\\midrule
AI / F & \EOneDAiFwdIntegrated & \EOneDAiJointFwd & \EOneDAiFwdGap{} [\EOneDAiFwdGapLo{}, \EOneDAiFwdGapHi{}]\\
AI / R & \EOneDAiRevIntegrated & \EOneDAiJointRev & \EOneDAiRevGap{} [\EOneDAiRevGapLo{}, \EOneDAiRevGapHi{}]\\
Shopping / F & \EOneDShopFwdIntegrated & \EOneDShopJointFwd & \EOneDShopFwdGap{} [\EOneDShopFwdGapLo{}, \EOneDShopFwdGapHi{}]\\
Shopping / R & \EOneDShopRevIntegrated & \EOneDShopJointRev & \EOneDShopRevGap{} [\EOneDShopRevGapLo{}, \EOneDShopRevGapHi{}]\\
News / F & \EOneDNewsFwdIntegrated & \EOneDNewsJointFwd & \EOneDNewsFwdGap{} [\EOneDNewsFwdGapLo{}, \EOneDNewsFwdGapHi{}]\\
News / R & \EOneDNewsRevIntegrated & \EOneDNewsJointRev & \EOneDNewsRevGap{} [\EOneDNewsRevGapLo{}, \EOneDNewsRevGapHi{}]\\
\bottomrule\end{tabular}
\end{table}

\paragraph{Coarser audits.}
Pooling activity-change bins lets the audit detect the richer-history failure, but only at high inert-timestamp rates. Detection shares are \EOneEAiFailFwd{} forward and \EOneEAiFailRev{} reverse for AI at \EOneEAiRate{} times the training rate, \EOneENewsFailFwd{} and \EOneENewsFailRev{} for news at \EOneENewsRate{} times, and \EOneEShopFailMin{}--\EOneEShopFailMax{} for shopping at \EOneEShopRate{} times. At \EOneEAiSaturationRate{} times, AI runs out of controls and decisions turn inconclusive. No source or variant met the prespecified standard $R^*$: at least \EOneERequiredDecisionShare{} equivalence under random and matched-variable timing, recovery of signed injections, and richer-history detection in both transfers.
None of the \HardCases{} hard cases is confirmed within tolerance. The richer-history rule $H$ fails for every source in both transfers, the learned rule $T$ is inconclusive throughout, and the withheld-input rule $V$ fails only for forward news.

\paragraph{Diagnostic precision.}
Diagnostics that look only at surrounding activity have Wilson half-widths above $0.05$, up to \BenchmarkWilsonMax{}, and draw-mean standard errors up to \BenchmarkMcseMax{}, too noisy to support a decision.

\paragraph{Running the audit.}\label{app:practitioner-audit}
To run Table~\ref{tab:protocol}, freeze the baseline and tolerance, and train cells and timing rules on separate people. Recalibrate the timing rate at each inert-timestamp rate and draw independent timestamps until controls run out. Normalize cell means, weight by timestamp mass, and report any lost weight. Bootstrap whole people with timestamps, bins, and fits fixed, and keep Monte Carlo precision separate from panel uncertainty. Use simultaneous intervals across rules, injections, rates, and variants; saturation or a missing arm counts as inconclusive.

\subsection{Selection strength and fitted timing models}\label{app:strength}

\begin{table*}[t]
\centering\footnotesize\setlength{\tabcolsep}{3pt}
\caption{Fitted timing coefficients, raw error and common-slope intervals ($\dagger\ddagger$), and decisions ($\dagger$). F/R denote forward/reverse transfers, W/X/I within tolerance/fails/inconclusive, and C/E/O the interval procedures below, with O overlays in parentheses; the common slope is the fitted strength $b$ in $\beta=b(1.25,1.25,1.25,-0.75)$, plotted as the forward fitted strength in Figure~\ref{fig:strength}; Table~\ref{tab:term-errors} reports procedure O.}
\label{tab:strength-slopes}\label{tab:fit-decisions}
\begin{tabular}{@{}lrrrrrrrr@{}}\toprule
Source / transfer & Rising & Search & Recency & Age & Error [\IntervalLevel]$\dagger\ddagger$ & Beyond $\pm\Tolerance$ & Common slope [\IntervalLevel]$\dagger\ddagger$ & C / E / O$\dagger$\\\midrule
AI / F & \TermAiFwdAccel & \TermAiFwdSearch & \TermAiFwdRecency & \TermAiFwdAge & \TermAiFwdError{} [\TermAiFwdErrorLo{}, \TermAiFwdErrorHi{}] & \TermAiFwdBeyond & \RichAiFwdSlope{} [\RichAiFwdSlopeLo{}, \RichAiFwdSlopeHi{}] & \TermAiFwdCoefficientDecision/\TermAiFwdEvaluationDecision/\EOneCAiFwdNullDecision(\EOneCAiFwdOverlayDecision)\\
AI / R & \TermAiRevAccel & \TermAiRevSearch & \TermAiRevRecency & \TermAiRevAge & \TermAiRevError{} [\TermAiRevErrorLo{}, \TermAiRevErrorHi{}] & \TermAiRevBeyond & \RichAiRevSlope{} [\RichAiRevSlopeLo{}, \RichAiRevSlopeHi{}] & \TermAiRevCoefficientDecision/\TermAiRevEvaluationDecision/\EOneCAiRevNullDecision(\EOneCAiRevOverlayDecision)\\
Shopping / F & \TermShopFwdAccel & \TermShopFwdSearch & \TermShopFwdRecency & \TermShopFwdAge & \TermShopFwdError{} [\TermShopFwdErrorLo{}, \TermShopFwdErrorHi{}] & \TermShopFwdBeyond & \RichShopFwdSlope{} [\RichShopFwdSlopeLo{}, \RichShopFwdSlopeHi{}] & \TermShopFwdCoefficientDecision/\TermShopFwdEvaluationDecision/\EOneCShopFwdNullDecision(\EOneCShopFwdOverlayDecision)\\
Shopping / R & \TermShopRevAccel & \TermShopRevSearch & \TermShopRevRecency & \TermShopRevAge & \TermShopRevError{} [\TermShopRevErrorLo{}, \TermShopRevErrorHi{}] & \TermShopRevBeyond & \RichShopRevSlope{} [\RichShopRevSlopeLo{}, \RichShopRevSlopeHi{}] & \TermShopRevCoefficientDecision/\TermShopRevEvaluationDecision/\EOneCShopRevNullDecision(\EOneCShopRevOverlayDecision)\\
News / F & \TermNewsFwdAccel & \TermNewsFwdSearch & \TermNewsFwdRecency & \TermNewsFwdAge & \TermNewsFwdError{} [\TermNewsFwdErrorLo{}, \TermNewsFwdErrorHi{}] & \TermNewsFwdBeyond & \RichNewsFwdSlope{} [\RichNewsFwdSlopeLo{}, \RichNewsFwdSlopeHi{}] & \TermNewsFwdCoefficientDecision/\TermNewsFwdEvaluationDecision/\EOneCNewsFwdNullDecision(\EOneCNewsFwdOverlayDecision)\\
News / R & \TermNewsRevAccel & \TermNewsRevSearch & \TermNewsRevRecency & \TermNewsRevAge & \TermNewsRevError{} [\TermNewsRevErrorLo{}, \TermNewsRevErrorHi{}] & \TermNewsRevBeyond & \RichNewsRevSlope{} [\RichNewsRevSlopeLo{}, \RichNewsRevSlopeHi{}] & \TermNewsRevCoefficientDecision/\TermNewsRevEvaluationDecision/\EOneCNewsRevNullDecision(\EOneCNewsRevOverlayDecision)\\
\bottomrule\end{tabular}
\end{table*}

The strength sweep in Figure~\ref{fig:strength} scales the unmatched history $g(H)=f_H(H)-f_C(H)$ by $s$:
\[
 e_s(H)=\operatorname{clip}_{[10^{-6},1-10^{-6}]}
          \operatorname{expit}(\alpha_s+f_C(H)+s\,g(H)).
\]
The matched-variable terms stay fixed, and $\alpha_s$ is recalibrated so the training event rate does not change. At $s=0$ matching has zero integrated error, and $s=1$ is the fixed richer-history rule. The grid is $0,0.1,\ldots,1,1.5,2$ with frozen bins.

\paragraph{Fitting the timing rule to real events.}
The fitted rule pairs each real event $j$ with every control $k$ in its matching cell $c$ and fits how the four added indicators $X$ shift the odds of being the event:
\[
 P(Z_j=1\mid Z_j+Z_k=1,H,c)
 =\operatorname{expit}\{\beta^\top(X_j-X_k)\}.
\]
Pairs are weighted by $w_jw_k/W_c$, and cell intercepts cancel; no coefficient approaches $\pm20$. Inert timestamps then follow the clipped, training-calibrated probabilities $\operatorname{expit}(\alpha_\beta+f_C+X^\top\beta)$, and errors are integrated exactly on fixed evaluation designs. For AI, the prior-search and rising-activity terms partly offset each other. Shopping's rising-activity coefficient changes sign between training months (\TermShopAccelChange{}; Table~\ref{tab:strength-slopes}), and its two point errors straddle the tolerance with both intervals inconclusive.

\paragraph{Three uncertainty procedures.}
Table~\ref{tab:fit-decisions} reports three procedures. C draws training coefficients (\RichSlopeDraws{} refits on training people) and holds the evaluation design fixed. E resamples evaluation people with coefficients fixed. O uses direct matching with simultaneous evaluation-person intervals conditional on the training fit. E and O widen by \CoverageK{}. A fourth, joint procedure resamples training and evaluation people together and refits the timing rule each time. Under it the AI error stays within tolerance at the 95\% level (forward [\JointAiFwdLo{}, \JointAiFwdHi{}], reverse [\JointAiRevLo{}, \JointAiRevHi{}]), but every fitted-rule cell is inconclusive at procedure O's simultaneous level; uncertainty in the training coefficients accounts for nearly all of the width.
In hourly units, the AI errors are \TermAiFwdResidual{} [\TermAiFwdResidualLo, \TermAiFwdResidualHi]$\dagger\ddagger$ forward and \TermAiRevResidual{} [\TermAiRevResidualLo, \TermAiRevResidualHi]$\dagger\ddagger$ reverse as fractions of the real-minus-eligible excess \LadderActiveRate{}; the targets and weights differ, so these are not causal shares. Model-based adjustment meets tolerance in \StrengthSweepPassSummary{} overlapping cases, which are not independent trials.

\begin{figure}[t]
\centering
\includegraphics[width=\columnwidth]{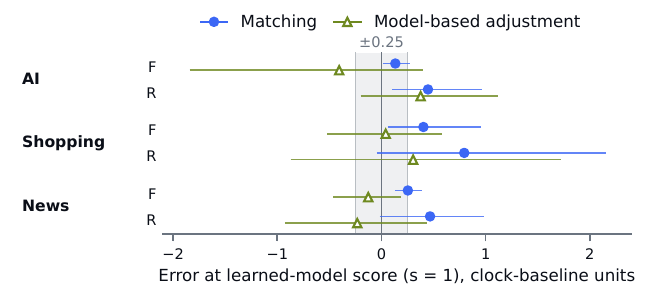}
\caption{\LearnedCaption}
\Description{Learned timing model strength curves show matching and model-based adjustment in both train/test directions, preserving the full intervals on a symmetric-log axis.}
\label{fig:learned}
\end{figure}

\paragraph{Other sweeps.}
A whole-score sweep scales $f_H$ by $s\in\{0,\allowbreak0.25,\allowbreak0.5,\allowbreak1,\allowbreak1.5,\allowbreak2\}$ without cell conditioning, and the learned model's score is scaled separately. Model-based adjustment is inconclusive in \StrengthAllDrInconclusive{} of these cases, with no confirmed failure. In the design of Table~\ref{tab:term-errors}, with marginal intervals, model-based adjustment under richer-history selection and both matching and adjustment under the learned timing model are inconclusive in all six cells.

\paragraph{Model-based adjustment on held-out people.}
Adjustment adds the count bin and the four indicators to the prior-history features; it never sees the generator's scores or probabilities. Held-out log loss selects the timing model and control squared error selects the outcome model.

\subsection{Additional adjustments}\label{app:expert}
We judge each remedy on separate people and in both transfers, by whether it removes error under a null and whether it keeps a real effect. Overlays $Y^\theta=Y+M\theta B$ with $\theta\in\{0,\pm\InjectedEffect\}$ inject known effects while holding controls and past features fixed.
Matching on richer history keeps too little of the target: minimum target weight is \AiRichWeight{} for AI, \ShopRichWeight{} for shopping, and \NewsRichWeight{} for news. A timing model removes the matched-variable rule's error but stays inconclusive under learned timing. Visit comparisons change the target, with worst richer-history errors of \AiVisitHard{} for AI, \ShopVisitHard{} for shopping, and \NewsVisitHard{} for news. After calibration, the worst AI null error is \AiExpertCalibrated{} [\AiExpertCalibratedLo, \AiExpertCalibratedHi]$\dagger$, and under overlap weights it is \AiExpertOverlap{} [\AiExpertOverlapLo, \AiExpertOverlapHi]$\dagger$, with wide intervals and a changed population.

\subsection{Duration predictability and post-event conditioning}\label{app:duration}
Predicting how long a stretch of activity lasts is near chance at the primary gap cutoff, and no other cutoff we tried (\MultiscaleGapCutoffs{} seconds) stands out. Adjusting for activity after the event creates its own bias~\cite{rosenbaum1984adjustment}, and that bias can pass inert-timestamp and constant-injection checks.
In the forward direction, for forecast-adjusted Google-host counts, adjusting only for pre-event activity leaves a worst error of \FutureFwdPastNullError{} clock-baseline units under both the null and a duration-changing effect. Adding post-event activity cuts the null error to \FutureFwdAllNullError{} but raises the duration-effect error to \FutureFwdAllDurationError{}. Reverse null checks fail, so neither variant qualifies, and covariates that do not depend on duration still fail.

\section{Identification arguments}\label{app:proofs}
With consistency for event $A$ and outcome $Y^a$, holding earlier treatment fixed,
\begin{align*}
 E[Y\mid A=1]-E[Y\mid A=0]
 &=E[Y^1-Y^0\mid A=1]\\
 &\quad+E[Y^0\mid A=1]-E[Y^0\mid A=0].
\end{align*}
The last line is selection within a person.

\paragraph{Generated inert timestamps.}\label{prop:j3-placebo}
A timestamp drawn from recorded history $H$ carries no information beyond $H$. Let $0<e(H)<1$, let $U$ be a fresh uniform independent of the outcome given $H$, and set $M=\mathrm I\{U\leq e(H)\}$. Then
\[
 P(M=1,Y\in C\mid H)=e(H)\,P(Y\in C\mid H),
\]
so $M\perp Y\mid H$ and correct adjustment for $H$ recovers zero. The oracle needs fixed weights and denominators; fitted or ratio estimators, and support chosen after assignment, lose this guarantee.

\paragraph{Residual-state formula.}\label{prop:j3-bias}
Let a binary state $L$ be hidden from history and assume, beyond inertness, $P(A=1\mid L=b,H)=p_b$ and $E[Y\mid A,L,H]=E[Y\mid L,H]$. Let $q=P(L=1\mid H)$, $\pi=p_0+(p_1-p_0)q$, and $\Delta=E[Y\mid L=1,H]-E[Y\mid L=0,H]$. Subtracting the Bayes posteriors $p_1q/\pi$ and $(1-p_1)q/(1-\pi)$, under overlap, gives
\[
 \operatorname{Bias}(H)=\Delta\frac{(p_1-p_0)q(1-q)}{\pi(1-\pi)}.
\]
The bias vanishes when $q\in\{0,1\}$, $p_1=p_0$, or $\Delta=0$.

\paragraph{Burst persistence.}\label{prop:burst-persistence}
Longer bursts raise the total bias over a fixed window. Let the state $L_s$ follow a Markov chain that turns active with probability $\alpha$ and quiet with probability $\beta$. With $\rho=1-\alpha-\beta$ and outcome-rate gap $\delta_\lambda=\lambda_1-\lambda_0$, iterating $E[L_{s+1}\mid L_s]=\alpha+\rho L_s$ gives a start-state difference of $\rho^s$, so the cumulative gap $\Delta_h=\delta_\lambda\sum_{s=1}^h\rho^s$ replaces $\Delta$. For $0<\rho<1$, $\delta_\lambda>0$, $p_1>p_0$, and $0<q<1$ the bias grows with persistence.

\paragraph{Observational equivalence.}\label{prop:nonid}\label{prop:j3-limit}
In model A, a persistent binary state raises Poisson outcomes, activity, and Bernoulli events; assignment changes nothing, but noisy history leaves events informative. Model B keeps A's laws for history and assignment and generates futures by sequential conditional quantiles of shared uniforms, so its contrast is causal. Both share one law for the logs, so every audit behaves the same.
\end{document}